\documentclass[11pt]{article}
\usepackage[T1]{fontenc}
\usepackage{graphicx}
\usepackage{fullpage}

\usepackage{amsthm,amssymb,amsmath}
\usepackage{latexsym}

\usepackage{color}

\usepackage{mathrsfs}

\usepackage{hyperref}
\usepackage{color}

\usepackage{url}
\usepackage{pstricks}
\usepackage{soul}
\usepackage{subcaption}
\usepackage{tikz,ifthen,calc}
\usepackage{stmaryrd}

\usepackage{mdwlist}
\usepackage{paralist}
\usepackage{xspace}

\usepackage{bbding}

\usepackage{cleveref}

\usetikzlibrary{positioning, fit}
\usetikzlibrary {arrows.meta}
\usetikzlibrary{positioning}
\usetikzlibrary{shapes}
\usetikzlibrary{shapes.symbols,patterns}
\usetikzlibrary{calc,through,backgrounds}
\usetikzlibrary{decorations.pathreplacing,calligraphy}
\usetikzlibrary{plotmarks}

\makeatletter
\newcommand{\xMapsto}[2][]{\ext@arrow 0599{\Mapstofill@}{#1}{#2}}
\def\Mapstofill@{\arrowfill@{\Mapstochar\Relbar}\Relbar\Rightarrow}
\makeatother

\newtheorem{lemma}{Lemma}[section]
\newtheorem{theorem}[lemma]{Theorem}

\newtheorem{observation}[lemma]{Observation}

\crefname{observation}{observation}{observations}

\newcommand{\integers}{\mathbb{Z}}

\newcommand{\set}[1]{\left\{ #1 \right\}}

\newcommand{\floor}[1]{\left\lfloor #1 \right\rfloor}

\def\FLOW{\mbox{\sf flow}}

\def\Ghalf{G^{1/2}}
\def\Vhalf{V^{1/2}}
\def\Ehalf{E^{1/2}}

\def\calC{\mathscr{C}}

\newcommand{\DRlong}{\textsc{Degree Realization}\xspace}
\newcommand{\DR}{\textsc{DR}\xspace}

\newcommand{\cliquelong}{\textsc{Maximum Clique}\xspace}
\newcommand{\clique}{\textsc{MC}\xspace}
\newcommand{\cliqueDRlong}{\cliquelong \textsc{Degree Realization}\xspace}
\newcommand{\cliqueDR}{\textsc{MC-DR}\xspace}

\newcommand{\islong}{\textsc{Maximum Independent Set}\xspace}
\newcommand{\is}{\textsc{MIS}\xspace}
\newcommand{\isDRlong}{\textsc{MIS Degree Realization}\xspace}
\newcommand{\isDR}{\textsc{MIS-DR}\xspace}

\newcommand{\vclong}{\textsc{Minimum Vertex Cover}\xspace}
\newcommand{\vc}{\textsc{MVC}\xspace}
\newcommand{\vcDRlong}{\textsc{MVC Degree Realization}\xspace}
\newcommand{\vcDR}{\textsc{MVC-DR}\xspace}

\newcommand{\matchlong}{\textsc{Maximum Matching}\xspace}
\newcommand{\match}{\textsc{MM}\xspace}
\newcommand{\matchDRlong}{\textsc{MM Degree Realization}\xspace}
\newcommand{\matchDR}{\textsc{MM-DR}\xspace}
\newcommand{\domsetlong}{\textsc{Minimum Dominating Set}\xspace}
\newcommand{\domset}{\textsc{MDS}\xspace}
\newcommand{\domsetDRlong}{\textsc{MDS Degree Realization}\xspace}
\newcommand{\domsetDR}{\textsc{MDS-DR}\xspace}

\newcommand{\InvMatch}{\textbf{Invert}\xspace}

\def\inline#1:{\par\vskip 3pt\noindent{\bf #1:}\hskip 10pt}
\def\midinline#1:{\par\noindent{\bf #1:}\hskip 10pt}
\def\dnsinline#1:{\par\vskip -7pt\noindent{\bf #1:}\hskip 10pt}

\long\def\commabs #1\commabsend{}

\begin{document}

\title{\textbf{Degree Realization with Maximum Matching}%
\thanks{An extended abstract was presented at the 37th International Workshop on Combinatorial Algorithms
(IWOCA) 2026. Supported by US-Israel BSF grant 2022205.}}

\author{%
Amotz Bar-Noy%
\thanks{City University of New York (CUNY), USA.} \\
{\tt\small amotz@sci.brooklyn.cuny.edu}
\and
Igor Kalinichev%
\thanks{Weizmann Institute of Science, Israel.} \\
{\tt\small igor.kalinichev@weizmann.ac.il}
\and
David Peleg%
\footnotemark[3] \\
{\tt\small david.peleg@weizmann.ac.il}
\and
Dror Rawitz%
\thanks{Bar Ilan University, Ramat-Gan, Israel.} \\
{\tt\small dror.rawitz@biu.ac.il}
}
\date{October 30, 2026}













\begin{titlepage}

\maketitle

\begin{abstract}
The \matchlong \DRlong
problem requires, for a sequence $d$ of $n$ positive integers, to compute a realizing graph 
(whose degree sequence corresponds to $d$) with the maximum possible matching.
We present a polynomial-time algorithm for this problem. 
The decision version of this problem requires to determine whether a sequence $d$ admits a realization with a matching of 
size $\nu$ (without computing such a realization). Recently, it was shown that it can be done by verifying a series 
of inequalities. We improve the existing characterization by reducing the number of inequalities.
\end{abstract}


\renewcommand{\thepage}{}
\end{titlepage}
\pagenumbering{arabic}


\section{Introduction}


The \DRlong (\DR) problem requires to decide, for a given non-increasing sequence $d=(d_1,\ldots,d_n)$ of positive integers, if there exists some $n$-vertex (simple undirected) graph $G=(V,E)$ such that $\deg_G(i) = d_i$ for every $i \in \set{1,\ldots,n}$, and to construct such a graph if it exists. 
Such a sequence $d$ is called \emph{graphic}.

Two central questions on realizations concern characterizations (i.e., necessary and sufficient conditions) for a sequence to be graphic, and efficient algorithms for finding realizing graphs.
A necessary and sufficient condition for a given sequence of integers to be
graphic 
(also implying an $O(n)$ decision algorithm) 
was presented by Erd\H{o}s and Gallai~\cite{EG60}.
Havel~\cite{Havel55} and Hakimi~\cite{Hakimi62} 
described an $O(\sum_i d_i)$-time 
algorithm that given a sequence $d$ of integers proves that the given sequence is not graphic or 
computes an $m$-edge graph realizing it, where 
$m ~=~ \frac{1}{2} \sum_{i=1}^n d_i$.

Here, we study the variant of the problem where, given a sequence with possibly many realizing graphs, 
the goal is to find a realizing graph that also optimizes some quality measure. 
Hereafter, we refer to such problems as \emph{optimized realization} problems.
Concretely, we consider the classical \matchlong (\match) problem. 
For a graph $G = (V,E)$, a matching is an edge set $M \subseteq E$, such that $e$ and $e'$ do not share a vertex, 
for every two distinct $e, e' \in M$. The size of the maximum matching in $G$ is denoted by
\[
\match(G) = \max \{|M| \mid M \text{ is a matching in } G\}
~.
\]
Given a graphic sequence $d$, we look for a realizing graph $G$ with the maximum possible matching. 
For example, for $d = (2,2,2,2,2,2)$, a realization $G_1$ consisting of a 6-vertex cycle
is better that a realization $G_2$ consisting of two triangles, since $\match(G_1) =3$ whereas $\match(G_2) = 2$.
The resulting optimized realization problem is formally defined as follows. 
Letting 
\[
\match(d) = \max \{\match(G) \mid G \text{ is a realization of } d\}
~,
\]
the \matchDRlong (\matchDR) problem requires finding a realizing graph $G$ attaining $\match(d)$.

Similarly, one may define optimized realization problems for 
other graph optimization problems, such as \cliquelong, \islong, \vclong, and \domsetlong.
Their optimal values on a given graph $G$ are defined as $\clique(G)$, $\is(G)$, $\vc(G)$, and $\domset(G)$, respectively. 
The corresponding optimized realization problems on degree sequences are named 
\cliqueDRlong (\cliqueDR), \isDRlong (\isDR), \vcDRlong (\vcDR), and \domsetDRlong (\domsetDR), and their optimal values on a given $d$ are denoted by 
$\clique(d)$, $\is(d)$, $\vc(d)$, and $\domset(d)$, respectively.

Interestingly, while the graph versions of most of these optimization problems are NP-hard, 
the optimized realization versions of \clique, \is, \vc, and \domset
have been shown to be polynomial-time solvable.
Rao~\cite{Rao79} gave a characterization of sequences realizable by a graph 
containing a clique $K_\ell$ of size $\ell$. This result was based on 
a phenomenon known as the \emph{prefix lemma}, which for MC says that
if $d$ has a realization containing $K_\ell$, then it has a realization such that 
the $\ell$ vertices of maximum degree induce a clique $K_\ell$.
This result implies a polynomial-time algorithm for \cliqueDR
(see, e.g., \cite{GHR18}).
%
%
Consequently,
it follows readily that polynomial-time algorithms also exist for \isDR and \vcDR.
Note that a prefix lemma also applies to \vc.
This prefix property is satisfied by \domset~{\cite{GHR18}}, 
and it was recently used to obtain a characterization and a realization algorithm
for \domsetDR~\cite{BKPR26-domset}.

As for \matchDR, the prefix property is satisfied also by \match~\cite{GJL99}.
Moreover, due to this property, $\match(d)$ can be computed efficiently using a result by Kundu{~\cite{Kundu73}} 
(see details below).
However, so far it was not known how to exploit this property in order to derive a polynomial-time algorithm for \matchDR.
This problem was handled for some restricted graph families in~\cite{BockRautenbach19}, 
but the general problem was left open.

Recently, Erd\H{o}s et al.~\cite{EKMT24} studied \matchDR. 
Given a sequence $d$ and an integer $\nu$, they presented an Erd\H{o}s-Gallai type characterization, based on a system of $O(n)$ inequalities,
which is satisfied if and only if $d$ has a realization with a matching of size $\nu$.
Applying these characterizations to a given degree sequence, the maximum size of a 
matching can be computed efficiently, although a specific realization is not provided.

\subsection{Our Results}
%
Here, we present a polynomial-time realization algorithm for \matchDR. 
Our algorithm makes use of a stronger version of an existing prefix lemma for general graphs~{\cite{GJL99}}, termed the \emph{inverted prefix lemma,}
established in the current paper. This lemma is analogous to the inverted prefix lemma established in~\cite{BockRautenbach19} for \emph{bipartite} realizations of two sequences 
(or a bi-sequence).
%
We also utilize an algorithm due to Fulkerson, Hoffman and McAndrew~{\cite{FHM1965}} for a special case of the $f$-factor problem (discussed later), 
 which is based on the following three phases:
 \begin{enumerate}[(1)]
\item 
Realization of the bi-sequence $(d,d)$ by a \emph{bipartite} graph $\hat{G}$,
\item Conversion of $\hat{G}$ to a half-integral general (non-bipartite) realization 
$G^\omega$ for $d$, where in a half-intergral realization an edge with weight half 
contributes half to the degree of its end-points.
\item Rounding $G^\omega$ to an integral general realization $G$.
\end{enumerate}

In addition, we show how the described method leads to the characterization recently proved by 
Erd\H{o}s et al.~\cite{EKMT24}. In fact, we show that one of the equations in the characterization is redundant.





\subsection{Related Work}
%
Several papers studied network realization with a given subgraph.
Kundu~\cite{Kundu73} showed that the sequences $d$ and 
the component-wise difference $d - d'$, such that $d'_i \in \set{k,k+1}$ and $d' \le d$ (component-wise), 
are graphic only if there exists a graph $G$ realizing $d$ that has a subgraph $G'$ realizing $d'$. 
This result can be used to decide whether $\match(d) \geq \ell$, for a given $\ell$, by checking 
$d' = (1^{2\ell},0^{n-2\ell})$. Furthermore, $\match(d)$ can be computed efficiently, 
since $\match(d)$ is is the maximum $\ell$, such that $d - (1^{2\ell},0^{n-2\ell})$ is realizable.
This result can also be used to decide whether a given sequence $d$ has a perfect matching by 
assigning $d'_i = 1$, for every $i$.
Kleitman and Wang~\cite{KleitmanW73} gave an algorithm for computing a realization of $d$ that contains a subgraph which realizes $d'$. 
Their algorithm can be used to compute a realization of $d$ that contains a perfect matching, 
if such a realization exists.
Extensions of the above result were presented in~\cite{KleitmanW73,Kundu74}.

Rao and Rao{~\cite{RaoRao72}} and Kundu{~\cite{Kundu73}} gave a characterization of 
sequences that can be realized by a Hamiltonian graph.
Chungphaisan~\cite{Chungphaisan78} gave an algorithm, that given a sequence $d$,
constructs a realization with a Hamiltonian cycle (or path), if such a realization exists.
Rao~\cite{Rao79} characterized sequences that can be realized by a graph containing $K_\ell$ (a clique of size $\ell$). This result was based on the Prefix Lemma for MC.
An alternative proof was given in~\cite{KezdyLehel96}, 
and a constructive proof was presented in~{\cite{Yin12}}.
It follows that \cliqueDR can be solved in polynomial time by 
checking if $d$ has a realization containing $K_\ell$, for 
$\ell \in \set{2,\ldots,n}$.
%

Gould, Jackson and Lehel~\cite{GJL99} extended Rao's result by showing that if $d$ 
has a realization containing $H$ as a subgraph (but not necessarily an induced sub-graph),
then there exists a realization of $d$ containing $H$ where the vertices of $H$ 
have the $|V(H)|$ largest degrees.
Yin~\cite{Yin11} used this to extend Rao's result~\cite{Rao79}  by giving a characterization of graphic sequences that 
contains a split graph $S_{r,s}$ composed of a clique is of size $r$ and an independent set of size $s$. 
(Observe that $S_{r,1}$ is a clique of size $r+1$.)

Gentner, Henning and Rautenbach~\cite{GHR18} 
studied \isDR and \domsetDR. They proved the prefix lemma for \domsetDR and gave a realization algorithm for \domsetDR on sequences with $d_1 = O(1)$, leaving the general case open. They also provided characterizations for \isDR and \domsetDR in forests.
Gentner, Henning and Rautenbach~{\cite{GHR16}} gave characterizations to realizations that minimize the maximum size of an independent set and that maximize the minimum size of a dominating set in forests.
Recently, Bar-Noy et al.~{\cite{BKPR26-domset}} gave a characterization and a realization algorithm 
for \domsetDR in general graphs.

Bock and Rautenbach~\cite{BockRautenbach19} studied \matchDR in trees and bipartite graphs, where the partition is given.
Their result on bipartite graphs is also based on a stronger version of the prefix lemma, which is essentially an \emph{inverted prefix lemma for bipartite graphs}.

Fulkerson, Hoffman and McAndrew~\cite{FHM1965} 
provided a characterization for the realizability of a sequence $d$
with respect to multi-graphs that satisfy the \emph{odd cycle condition} (discussed later).
Building on the result of~\cite{FHM1965}, 
Kundu~\cite{Kundu74a} presented conditions for the factorization of graphs that satisfy the odd cycle condition.
Anstee relied on the technique of~\cite{FHM1965} to obtain an algorithmic proof of the $f$-factor theorem~\cite{Anstee85} and a simplified characterization for the existence of a $(g,f)$-factor for special cases~\cite{Anstee90}.
%
The result of Yin~{\cite{Yin11}}, mentioned above, was also based on~\cite{FHM1965}.
%
The results for \domsetDR~\cite{BKPR26-domset} are also based on this technique.


\section{Inverted Prefix Lemma}

In this section, we establish the inverted prefix lemma.
Recall that Gould, Jacobson and Lehel~{\cite{GJL99}} gave a general prefix lemma that applies 
to any subgraph. In the following lemma we provide the special case of a matching.

\begin{lemma}[Arbitrary Prefix Lemma for \match~\cite{GJL99}]
\label[lemma]{cl: Gould for MM}
If a sequence $d$ has a realization with a matching of size $\nu$, then $d$ has a realization with a matching of the same size such that $2 \nu$ vertices participating in the matching have the highest $2\nu$ degrees in $d$, i.e., $d_1, d_2, \ldots, d_{2\nu}$. 
\end{lemma}

For our purposes, however, we need a stronger type of prefix lemma, similar to the one used in~\cite{BockRautenbach19} for solving MM-DR over bipartite graphs. Towards proving this stronger claim, we introduce some notation. 
Consider a realization $G = (V, E)$ of a degree sequence $d$, 
with four vertices $x, y, u, v \in V$ such that $(x,u), (y,v) \in E$ and $(x,y), (u, v) \notin E$. 
Then the \emph{FLIP operation} transforms $G$ into a graph $G' = (V,E')$ by replacing the former 
two edges with the latter two, i.e., setting 
\[
E' = E \cup \{(x, y), (u, v)\} \setminus \{(x, u), (y, v)\}
~. 
\]
This operation preserves the degrees of individual vertices, but could lead to non-isomorphic realization. 
It was widely studied and used in different contexts, see~\cite{Fulkerson1960,EKM13,Barrus16}.
Finally, define $[i,j] = \{i,\ldots,j\}$, for $i \leq j$ and 
$[i] = [1,i] = \{1, \ldots, i\}$ for an integer $i$.
\begin{lemma}
\label[lemma]{lem: MM structure}
Let $G = (V, E)$ be a realization of $d$ with $V = \{v_1, \ldots, v_n\}$, 
such that $\deg_G(v_i) = d_i$ for every $i \in [n]$. 
If there is a matching $M  \subseteq E$ covering the first $2\nu$ vertices $v_1,\ldots,v_{2\nu}$, 
then there is a realization $G' = (V, E')$ of $d$ with the "inverted" matching 
\[
M^{inv} = \{(v_i, v_{2\nu-i+1}) \mid i \in [\nu]\}
~,
\]
such that $M^{inv} \subseteq E'$.
\end{lemma}
\begin{proof}
Call the pair $(G,M)$ \emph{valid} if $G=(V,E)$ is defined on $V = \{v_1, \ldots, v_n\}$ with $\deg(v_i) = d_i$, $M\subseteq E$, and $M$ is on $\{v_1, v_2, \ldots, v_{2\nu}\}$. 
Consider a valid pair $(G, M)$ and let 
\[
I(M) = \set{ i \mid (v_i, v_{2\nu-i+1}) \notin M, ~i \in [\nu]}
\]
and $f(M) = \min I(M)$. 
If $I(M)=\emptyset$, then $M=M^{inv}$, and we are done.
%
    
We describe procedure \InvMatch that given a valid pair $(G,M)$ violating the claim produces a new valid pair $(G', M')$, such that either $(G',M')$ satisfies the claim or $f(M') > f(M)$. Repeatedly applying \InvMatch
gradually modifies $(G,M)$ while preserving validity and strictly increasing $f(M)$. Since $f$ is bounded above by $\nu$, the process terminates in at most $\nu$ steps and the obtained pair $(G', M')$ satisfies the claim. 

Given a pair $(G,M)$ not satisfying the claim, let $i = f(M)$ and $j = 2\nu -i +1$. 
Since $(v_i,v_j) \notin M$, there are edges $(v_i,v_x), (v_y,v_j) \in M$, and the involved vertices satisfy $\deg_G(v_i) \ge \deg_G(v_y)$ and $\deg_G(v_j) \le \deg_G(v_x)$ by the definition of $f$. 
Let 
\[
M' = M \cup \{(v_i,v_j), (v_x, v_y)\}\setminus \{(v_i,v_x), (v_y,v_j)\}
~. 
\]
If there is a realization $G'$ such that $M' \subseteq G'$, then the pair $(G', M')$ is valid and either satisfies the claim or increases $f$. Proceed according to one of the following cases to obtain $G'$. See \Cref{fig: MM prefix cases}.
%
\begin{enumerate}[{Case} 1:]
\item 
$(v_i, v_j), (v_x, v_y) \in E(G)$. 

Then simply take $G' = G$.

\item $(v_i, v_j), (v_x, v_y) \notin E(G)$. 

Then perform a FLIP operation on $G$ replacing $(v_i,v_x), (v_y,v_j)$ with $(v_i, v_j), (v_x, v_y)$. $G'$ obtained in this way satisfies $M' \subseteq E(G')$.

\item $(v_i, v_j) \in E(G)$ and $(v_x, v_y) \notin E(G)$. 

Then there is a vertex $v_z \in V(G) \setminus \{v_i,v_j,v_x,v_y\}$, such that $(v_x, v_z) \in E(G)$ and $(v_j, v_z) \notin E(G)$. 
Indeed, $\deg_G(v_x) \ge \deg_G(v_j)$ and $v_y$ is connected to $v_j$, but not to $v_x$, so there exists $v_z$ compensating for this. 
Obtain $G'$ by doing a FLIP replacing $(v_x, v_z), (v_j,v_y)$ with $(v_x,v_y), (v_j, v_z)$.

\item $(v_i, v_j) \notin E(G)$ and $(v_x, v_y) \in E(G)$.

Then there is a vertex $v_z \in V(G) \setminus \{v_i,v_j,v_x,v_y\}$, such that $(v_i, v_z) \in E(G)$ and $(v_y, v_z) \notin E(G)$. 
Similarly to the previous case, this is because $\deg_G(v_i) \ge \deg_G(v_y)$. 
Obtain $G'$ by doing a FLIP replacing $(v_i, v_z), (v_j,v_y)$ with $(v_i,v_j), (v_y, v_z)$.
\end{enumerate}
The lemma follows.
\end{proof}

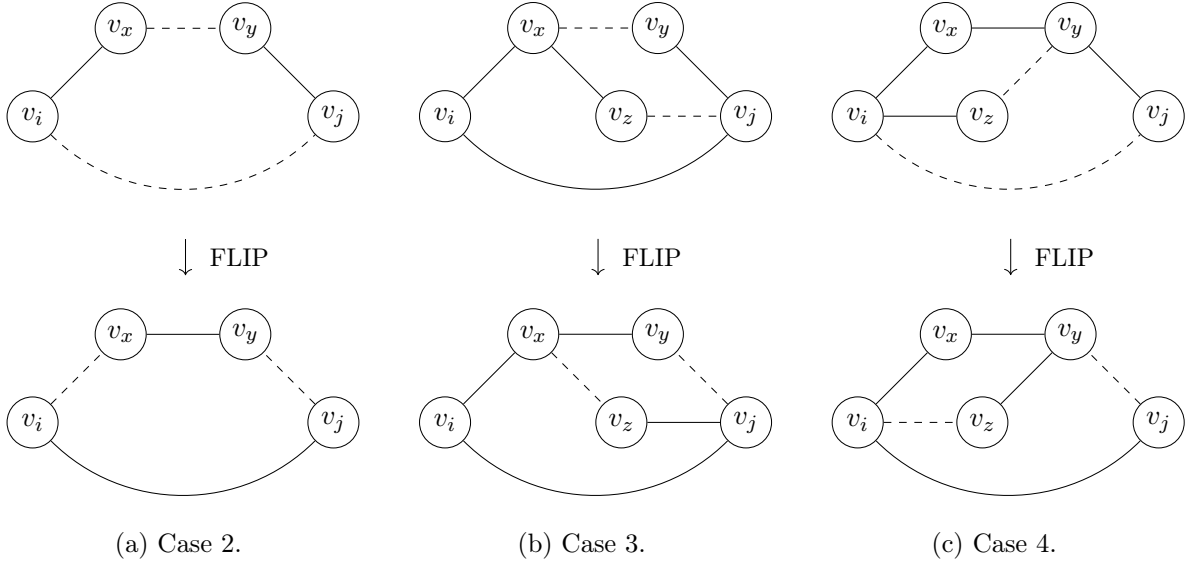
\begin{figure}[t]
\centering
\begin{tabular}{c@{\hspace{5pt}}c@{\hspace{5pt}}c}
\begin{subfigure}{.305\textwidth}
    \centering
    \resizebox{\textwidth}{!}{
    \tikzset{
    myptr/.style={-{Stealth[scale=1.5]}},
    }
    \begin{tikzpicture}[
        mynode/.style={circle, draw=black, minimum size=20pt, inner sep=0pt}
    ]
    
    \node[mynode] (i) {$v_i$};
    \node[mynode, above right = 1cm of i] (x) {$v_x$};
    \node[mynode, right =1cm of x] (y) {$v_y$};
    \node[mynode, below right = 1cm of y] (j) {$v_j$};
    \node[right = 0.4cm of x] (sup) {};
    
    \draw[dashed] (x) -- (y);
    \draw[dashed] (i) to[out=-45, in=-135] (j);
    \draw (i) -- (x);
    \draw (j) -- (y) ;
    
    \node[below=2.5cm of sup] (m1) {};
    \node[below=0.5cm of m1] (m2) {};
    \draw[->]        (m1)   -- (m2) node[pos=0.5, right=0.2cm] {\small FLIP};
    
    \node[mynode, below = 3.5cm of i] (i1) {$v_i$};
    \node[mynode, above right = 1cm of i1] (x1) {$v_x$};
    \node[mynode, right =1cm of x1] (y1) {$v_y$};
    \node[mynode, below right = 1cm of y1] (j1) {$v_j$};
    
    \draw (x1) -- (y1);
    \draw (i1)  to[out=-45, in=-135] (j1);
    \draw[dashed] (i1) -- (x1);
    \draw[dashed] (j1) -- (y1);
    
    \end{tikzpicture}
    }
    \caption{Case 2.}
    \label{fig: MM prefix 2}
\end{subfigure}
&
\begin{subfigure}{.32\textwidth}
    \centering
    \resizebox{\textwidth}{!}{
    \tikzset{
    myptr/.style={-{Stealth[scale=1.5]}},
    }
    
    \definecolor{b}{rgb}{0.0, 0, 1}
    \definecolor{gr}{rgb}{0.7, 0.7, 0.7}
    
    \begin{tikzpicture}[
        mynode/.style={circle, draw=black, minimum size=20pt, inner sep=0pt}
    ]

    \node[mynode] (i) {$v_i$};
    \node[mynode, above right = 1cm of i] (x) {$v_x$};
    \node[mynode, right =1cm of x] (y) {$v_y$};
    \node[mynode, below right = 1cm of y] (j) {$v_j$};
    \node[below = 0.2cm of x] (sup) {};
    \node[mynode, below right = 1cm of x] (z) {$v_z$};
    \node[right = 0.4cm of x] (sup) {};
    
    \draw[dashed] (x) -- (y);
    \draw (i) to[out=-45,in=-135] (j);
    \draw (i) -- (x);
    \draw (j) -- (y);
    \draw (x) -- (z);
    \draw[dashed] (z) -- (j);

    \node[below=2.5cm of sup] (m1) {};
    \node[below=0.5cm of m1] (m2) {};
    \draw[->]  (m1)   -- (m2) node[pos=0.5, right=0.2cm] {\small FLIP};
    
    \node[mynode, below = 3.5cm of i] (i1) {$v_i$};
    \node[mynode, above right = 1cm of i1] (x1) {$v_x$};
    \node[mynode, right =1cm of x1] (y1) {$v_y$};
    \node[mynode, below right = 1cm of y1] (j1) {$v_j$};
    \node[mynode, below right = 1cm of x1] (z1) {$v_z$};
    
    \draw (x1) -- (y1);
    \draw (i1) to[out=-45,in=-135] (j1);
    \draw (i1) -- (x1);
    \draw[dashed] (j1) -- (y1);
    \draw[dashed] (x1) -- (z1);
    \draw (z1) -- (j1);
    
    \end{tikzpicture}
    }
    \caption{Case 3.}
    \label{fig: MM prefix 3}
    \end{subfigure}
&
    \begin{subfigure}{.32\textwidth}
    \centering
    \resizebox{\textwidth}{!}{
    \tikzset{
    myptr/.style={-{Stealth[scale=1.5]}},
    }
    
    \definecolor{b}{rgb}{0.0, 0, 1}
    \definecolor{gr}{rgb}{0.7, 0.7, 0.7}
    
    \begin{tikzpicture}[
        mynode/.style={circle, draw=black, minimum size=20pt, inner sep=0pt}
    ]

    \node[mynode] (i) {$v_i$};
    \node[mynode, above right = 1cm of i] (x) {$v_x$};
    \node[mynode, right =1cm of x] (y) {$v_y$};
    \node[mynode, below right = 1cm of y] (j) {$v_j$};
    \node[right = 0.2cm of x] (sup) {};
    \node[mynode, below left = 1cm of y] (z) {$v_z$};
    \node[right = 0.4cm of x] (sup) {};
    
    \draw (x) -- (y);
    \draw[dashed] (i) to[out=-45,in=-135] (j);
    \draw (i) -- (x);
    \draw (j) -- (y);
    \draw[dashed] (y) -- (z);
    \draw (i) -- (z);

    \node[below=2.5cm of sup] (m1) {};
    \node[below=0.5cm of m1] (m2) {};
    \draw[->]        (m1)   -- (m2) node[pos=0.5, right=0.2cm] {\small FLIP};
    
    \node[mynode, below = 3.5cm of i] (i1) {$v_i$};
    \node[mynode, above right = 1cm of i1] (x1) {$v_x$};
    \node[mynode, right =1cm of x1] (y1) {$v_y$};
    \node[mynode, below right = 1cm of y1] (j1) {$v_j$};
    \node[mynode, below left = 1cm of y1] (z1) {$v_z$};
    
    \draw (x1) -- (y1);
    \draw (i1) to[out=-45,in=-135] (j1);
    \draw (i1) -- (x1);
    \draw[dashed] (j1) -- (y1);
    \draw (y1) -- (z1);
    \draw[dashed] (i1) -- (z1);
    
    \end{tikzpicture}
    }
    \caption{Case 4.}
    \label{fig: MM prefix 4}
    \end{subfigure}
\end{tabular}
\caption{Illustration of the cases in the proof of \Cref{lem: MM structure}.
All the illustrating figures maintain the convention that 
solid lines represent edges existing in the graph and dashed lines represent edges not in the graph.
}
\label{fig: MM prefix cases}
\end{figure}

Combining \Cref{cl: Gould for MM} and \Cref{lem: MM structure} gives a more specific Prefix Lemma for \matchlong.

\begin{lemma}[Inverted Prefix Lemma for \match]
\label[lemma]{lem: MM prefix}
If a sequence $d$ has a realization $G = (V, E)$ where $V = \{v_1,\ldots, v_n\}$ such that $\deg_G(v_i) = d_i$ for every $i \in [n]$,
with a matching $M$ of size $\nu$, then $d$ has a realization $G' = (V, E')$ with the 
matching $M' = \set{(v_i, v_{2\nu-i+1}) \mid i \in [\nu]}$.
%
\end{lemma}


\section{Algorithm for Realization with Maximum Matching}

Fulkerson, Hoffman and McAndrew~\cite{FHM1965} presented an algorithm for a special case of the $f$-factor problem.
More specifically, they considered the following degree realization problem.
Let $d$ be a sequence, and let $c \in (\integers_+)^{n \times n}$ be a symmetric integral matrix.
Intuitively, one can think of $c$ as representing a multigraph $K_c$ over $V = \set{1, \ldots, n}$ with $c_{i,j}$ parallel edges between $i$ and $j$ for every $1\le i,j\le n$.
Fulkerson et al.~\cite{FHM1965} studied 
the problem of finding a multigraph $H = (V,E)$, such that $\deg_H(i) = d_i$ and the number of parallel edges between $i$ and $j$ is bounded above by $c_{i,j}$.
Intuitively, the sought realizing multigraph $H$ needs to be a sub-multigraph, or factor, of $K_c$.
We refer to such a realization as a \emph{realization of $d$ with respect to $c$}.

Let $G_c = (V,E_c)$ be a graph derived from $c$, where 
\[
E_c = \set{ (i,j) \mid i,j \in V \text{ and } c_{i,j} > 0}
~.
\]
That is, $G_c$ is obtained from $K_c$ by replacing all multiedges by edges.

The matrix $c$ is said to satisfy the \emph{odd-cycle} condition if and only if the distance between any two disjoint odd cycles in the derived graph $G_c$ is at most one.
In other words, any two disjoint odd cycles in $G_c$ must be connected by an edge.

Fulkerson et al.~\cite{FHM1965} presented a characterization for the realizability of a 
sequence $d$ with respect to $c$ that satisfies the odd-cycle condition.
Their proof is constructive in the sense that it implies a realization algorithm for instances,
where $c$ satisfies the odd-cycle condition. 
The details of this algorithm for the specific context of realization with maximum matching are given in \Cref{s:MM-DR alg}.

\begin{theorem}[\cite{FHM1965}]
\label[theorem]{thm:fulkerson_alg}
Let $d$ be a sequence of size $n$, and let $c \in (\integers_+)^{n \times n}$ be a symmetric integral matrix.
There exists a polynomial-time algorithm that computes a realization of $d$ with respect to $c$ or determine that there is no such realization, assuming that $c$ satisfies the odd-cycle condition.
\end{theorem}


We show that we can use \Cref{thm:fulkerson_alg} to obtain a realization algorithm 
for \matchDR with a matching of size $\nu = \match(d)$.

Given a sequence $d$, our first step is to compute $\match(d)$ 
using the Erd\H{o}s-Gallai type characterization of Erd\H{o}s et al.~\cite{EKMT24}.

Next, Based on the Inverted Prefix Lemma (\Cref{lem: MM prefix}), we construct an instance $d'$ and $c'$
as follows:
\begin{align*}
d'_i & = 
\begin{cases}
d_i - 1 & i \leq 2\nu, \\
d_i     & i > 2\nu.
\end{cases}
& \text{and} &&
c_{i,j}' & =
\begin{cases}
0 & i = j, \\
0 & i + j = 2\nu+1 \\
1 & \text{otherwise}.
\end{cases}
\end{align*}

\begin{observation}
There is a realization of $d$ with a matching of size $\nu$
if and only if there is a realization of $d'$ with respect to $c'$.
\end{observation}

\begin{observation}[Odd-cycle condition]
\label[observation]{obs:odd-cycle-condition}
$c'$ satisfies the odd-cycle condition.
That is, for any pair of disjoint odd cycles $C_1, C_2$ in $G_{c'}$, there exist $i \in C_1$ and $j \in C_2$, 
such that $(i,j) \in E_{c'}$.
\end{observation}
\begin{proof}
Let $i \in C_1$, $j' \in C_2$ be arbitrary vertices.
If $i + j' = 2\nu + 1$, then let $j \neq j'$ be any other vertex in $C_2$. 
Otherwise, let $j = j'$.
\end{proof}

This implies the following.

\begin{theorem}
\label[theorem]{thm:MM}
There exists a polynomial-time algorithm for constructing a realization of a 
given graphic sequence $d$ that also has a matching of maximum size 
(among all possible realizations of $d$).
\end{theorem}


\section{Detailed Realization Algorithm}
\label{s:MM-DR alg}


In this section we present the full details of the realization algorithm that was mentioned in the previous section.
The algorithm proceeds in several steps, presented in the coming subsections.
First, the algorithm computes $\nu = \match(d)$ using 
the Arbitrary Prefix Lemma (\Cref{cl: Gould for MM}) and Kundu's result{~\cite{Kundu73}}.
Next, the \matchDR problem over general graph is reduced to the same problem over bipartite graphs, by applying the realization algorithm of 
Fulkerson et al.~\cite{FHM1965} while ensuring preservation of a matching.
Then, \matchDR over bipartite graphs is reduced to a maximum flow problem. 
To be more specific, given the size $\nu$ of the maximum matching and a degree sequence pair $(d,d)$, 
we construct a bipartite flow graph $G_{d,\nu}$, 
such that if the maximum flow attained in $G_{d,\nu}$ equals $\sum_{i=1}^n d_i - 2\nu$, 
then it corresponds to a realization $\hat{G}$ of $(d,d)$ with a matching of size $2\nu$. 
(This is where we make use of the Inverted Prefix {\Cref{lem: MM prefix}}.)
The next step is to convert $\hat{G}$ to a half integral general (non-bipartite) realization $G^\omega$ for $d$.
Finally, $G^\omega$ is rounded to an integral general realization $G$.



\subsection{Computation of $\match(d)$}


As mentioned before, given a candidate value $\nu$, one can check whether $\match(d) \geq \nu$
by using the result of Kundu~{\cite{Kundu73}} with $d' = (1^{2\nu},0^{n-2\nu})$. That can be done in 
$O(n)$ time using{~\cite{EG60}}. Hence, $\match(d)$ can be computed in $O(n \log n)$ time using binary search.


\subsection{Reduction to a Bipartite Sequence Pair}

%
Given a sequence $d$,
a bipartite graph $\hat{G} = (V, W, \hat{E})$, where $V = \set{v_1, \ldots, v_n}$ and
$W = \set{w_1, \ldots, w_n}$, is a \emph{$\nu$-matched} realization for the bi-sequence,
$(d, d)$ if it satisfies the properties:
\begin{enumerate}[(M1)]
\item $\hat{G}$ realizes the sequence pair $(d,d)$.
\item $(v_i, w_i) \notin \hat{E}$ for every $i \in [n]$.
\item $\hat{M} = \set{(v_i, w_{2\nu-i+1}) \mid i \in [2\nu]} \subseteq \hat{E}$.
\end{enumerate}

First, we describe an algorithm that given a $\nu$-matched realization $\hat{G}$ of $(d,d)$, produces a graph $G$ realizing $d$ with a matching of size $\nu$.
%
%
%

\begin{enumerate}
\item \label{step: 1}
Compute a half-integral solution.
\begin{enumerate}
\item For all $i, j \in [n]$, let 
\begin{align*}
y_{ij} & 
= \begin{cases}
    1, ~ \{v_i, w_j\} \in \hat{E}, \\
    0, ~ \text{otherwise,}
\end{cases}
&
\text{and}
&& 
\omega(i,j) = \frac{1}{2}(y_{ij} + y_{ji}).
\end{align*}

\item 
Define a weighted graph $G^\omega = (V^\omega, E^\omega,\omega)$ with vertex set $V^\omega = [n]$ and 
an edge $e=(i, j)$ of weight $\omega(e)=$ $\omega(i,j)$ for every $i,j \in V^\omega$. 
%
Clearly, $\omega$ is half-intergral.
\item Let $M^\omega= \{(i, 2\nu-i+1) \mid ~i \in [\nu]\}$ be the inverted matching in $G^\omega$. 
According to (M3)
\begin{align}
\label{eq: matching edges}
    \omega(e) = 1, ~\text{for every $e \in M^\omega$}.
\end{align}

\item Define the \emph{weighted degree} of a vertex $i \in V^\omega$ to be 
\[
d^\omega(i) = \sum_{j\in V^\omega} \omega(i, j)
~. 
\]
Note that $G^\omega$ realizes $d$ in the \emph{weighted} sense, namely,
\begin{align*}
d^\omega(i) 
& = \sum_{j\in V^\omega} \omega(i, j) = \frac{1}{2} \left(\sum_{j=1}^n y_{ij} + \sum_{j=1}^n y_{ji}\right) = d_i, 
& \mbox{for any}~ i \in V^\omega.
\end{align*}
\end{enumerate}

\item
\label{step: 2} 
Preparing for discarding non-integral weights while keeping the degrees unchanged.
\\
Construct a graph $\Ghalf = (\Vhalf,\Ehalf)$ obtained by removing from $G^\omega$ the edges of integral weight and keeping only those of weight $1/2$.
Formally, $\Vhalf = V^\omega$ and $\Ehalf = \{e \in E^\omega \mid \omega(e) = 1/2 \}$.

\medskip\noindent
In later stages of the construction, whenever $G^\omega$ is modified (by changing the weight of some edge $e$ from $1/2$ to 0 or 1), $\Ghalf$ is modified accordingly (by removing the edge $e$).

\item
\label{step: 3}
Partition into cycles.
\begin{enumerate}
\item 
Partition the edge set of $\Ghalf$ into disjoint (not necessarily simple) cycles each covering an entire connected component of the graph. Since $\Ghalf$ is an even graph by \Cref{obs: G1/2 is even graph for MM}, this can be done in polynomial number of steps.
\item
Let $\calC$ be a set of cycles in the partition.
\item
Let $\calC^{even} \gets \{C\in\calC \mid C ~\mbox{is of even length}\}$,
~~~~~
$\calC^{odd} \gets \calC \setminus \calC^{even}$.
\end{enumerate}

\commabs
The next two steps describe how to modify $G^\omega$ to obtain a realization $G$ of $d$.
\commabsend

\item 
\label{step: 4}
Eliminate even cycles.

\noindent
For every cycle $C \in \calC^{even}$, do the following.
\begin{enumerate}
\item 
Traverse $C$ starting from an arbitrary vertex $x \in C$ and continuing along the cycle until returning to $x$. Denote  the resulting sequence of edges by $E(C)=(e_1, e_2, \ldots, e_\ell)$. 
\item
Increase (resp., decrease) edge weights on even (resp., odd) positions in the sequence $E(C)$ by $1/2$.
That is, for every $i \in [\ell]$,
$\omega(e_i)$ is set to 1 if $i$ is even and 0 otherwise. 
\commabs
\begin{align*}
\omega(e_i) \gets \begin{cases}
        1, ~ \text{$i$ is even}, \\
        0, ~ \text{$i$ is odd}.
    \end{cases} 
\end{align*}
\commabsend
(This does not change the weighted degrees in $G^\omega$ and does not affect other cycles in $\calC$ or edges in $M^\omega$.)
\end{enumerate}

\item 
\label{step: 5}
Eliminate odd cycles.

\noindent
Arrange the cycles in $\calC^{odd}$ in pairs 
(by \Cref{obs: even number of odd cycles for MM} their number is even).
%
For every pair of cycles $(C,C')$ choose vertices $x \in C$, $y \in C'$ according to \Cref{obs: right x and y for MM}. Then do the following.
\begin{enumerate}
\item 
Starting from $x \in C$, traverse $C$ via edges 
$E(C)=(e_1, e_2, \ldots, e_\ell)$.
\\
Starting from vertex $y \in C'$, traverse $C'$ via edges 
$E(C')=(e'_1, \ldots, e'_k)$.
\item 
Let $\xi = \omega(x,y)$ 
(by \Cref{obs: right x and y for MM}, this weight must be either $0$ or $1$).
\item Set $\omega(x, y) \gets 1-\xi$.
\item For every $i \in [\ell]$ and $j \in [k]$, 
      modify the edge weights in the cycles $C$ and $C'$ as follows:
\begin{align*}
\omega(e_i) & \gets 
    \begin{cases}
        1-\xi, ~ \text{$i$ is even}, \\
        \xi, ~ \text{$i$ is odd},
    \end{cases} 
&
\omega(e'_j) & \gets 
    \begin{cases}
        1-\xi, ~ \text{$j$ is even}, \\
        \xi, ~ \text{$j$ is odd}.
    \end{cases} 
\end{align*}
\end{enumerate}
By \Cref{obs: G^omega realizes d for MM}, at this stage each edge in $G^\omega$ 
has weight 1 or 0, $G^\omega$ realizes $d$, and Eq.~\eqref{eq: matching edges} holds.

\item 
\label{step: 6}
Generate the output $G$.

\noindent
Transform $G^\omega$ into a simple graph $G$, with an edge $(i,j)$ whenever $\omega(i,j) = 1$ in $G^\omega$. Note that $M = M^\omega$ forms a matching in $G$ by Eq.~\eqref{eq: matching edges}.
\end{enumerate}


The algorithm is justified by the following Observations.

\begin{observation} 
\label[observation]{obs: G1/2 is even graph for MM}
The graph $G^{1/2}$ is even
(i.e., all its vertex degrees are even).
\end{observation}

A connected even graph has an Euler cycle, implying the following.

\begin{observation} 
\label[observation]{obs: even number of odd cycles for MM}
The number of cycles in $\calC^{odd}$ is even.
\end{observation}
\begin{proof}
Observe that
$\sum_{e \in E^\omega} \omega(e) ~=~ \frac{1}{2} \sum_{i=1}^n d_i = m,$
where $m$ is the number of edges, which is an integer.
Therefore, the number of edges with weight $1/2$ must be even. 
Since the cycles in $\calC$ cover all of $\Ehalf$ and are disjoint, the observation follows.
\end{proof}

The following observation is implied by \Cref{obs:odd-cycle-condition}.

\begin{observation}
\label[observation]{obs: right x and y for MM}
For any pair of disjoint odd cycles $C, C' \in \calC^{odd}$, there exist $x \in C$ and $y \in C'$, such that $(x,y) \notin M^\omega$ and $\omega(x, y) \neq 1/2$.
\end{observation}
\begin{proof}
%
\Cref{obs:odd-cycle-condition} implies that for any pair of disjoint odd cycles $C, C' \in \calC^{odd}$, there exist $x \in C$ and $y \in C'$, such that $(x,y) \notin M^\omega$.

Assume, towards contradiction, that $\omega(x,y)=1/2$. Then the edge $(x,y)$ belongs to $G^{1/2}$, implying that $x$ and $y$ belong to the same connected component in $G^{1/2}$. But this contradicts the fact that $x\in C$ and $y\in C'$ where $C$ and $C'$ cover different connected components of $\Ghalf$. 
\end{proof}

\begin{observation} 
\label[observation]{obs: G^omega realizes d for MM}
At the end of Step \ref{step: 5}, each edge in 
$G^\omega$ has weight 1 or 0, $G^\omega$ 
realizes $d$ and Eq. \eqref{eq: matching edges} holds.
\end{observation}
\begin{proof}
After Step \ref{step: 5}, $G^\omega$ contains no more edges of weight $1/2$. Moreover, neither of the modifications performed in Step \ref{step: 5} change the weighted degrees in $G^\omega$, affect other cycles in $\calC$ or edges in $M^\omega$. The observation follows.
\end{proof}

The reduction correctness follows from the next lemma.
\begin{lemma}
\label[lemma]{lem: MM equivalence}
There exists a $\nu$-matched realization of the bi-sequence $(d,d)$ if and only if $\match(d) \ge \nu$.
\end{lemma}
\begin{proof}
If there is a $\nu$-matched realization of $(d,d)$, then the aforementioned algorithm produces a realization of $d$ with a matching of size $\nu$.
    
For the other direction, given a realization $G = (V, E)$ of $d$ with vertices $V = [n]$ and 
a matching $M \subseteq E$ of size $\nu$, there is a realization $G' = (V, E')$ with 
a matching $M' = \{(v_i, v_{2\nu-i+1}) \mid i \in [\nu] \}$ by the Inverted Prefix Lemma (\Cref{lem: MM prefix}). 
Consider the following graph $\hat{G} = (\hat{V}, \hat{W}, \hat{E})$ with $\hat{V} = \{v_1, \ldots, v_n\}$, $\hat{W} = \{w_1, \ldots, w_n\}$ and $(v_i, w_j) \in \hat{E}$ if and only if $(i, j) \in E'$. 
Clearly, $\hat{G}$ realizes $(d,d)$ and $(v_i,w_i) \notin \hat{E}$ for every $i \in [n]$, 
so it satisfies (M1) and (M2). Moreover, the set $\hat{M} = \{(v_i, w_j) \mid (i,j) \in M'\}$ 
is a matching of size $2 \nu$ in $\hat{G}$ satisfying $(M3)$. Indeed, since $M'$ is a matching, 
the edges in $\hat{M}$ form a matching too and each edge $(i,j) \in M'$ leads to two edges 
$(v_i, w_j), (v_j, w_i) \in \hat{M}$. The lemma follows.
\end{proof}

\subsection{Reduction to Flow}

According to the previous section (in particular, \Cref{lem: MM equivalence}) given a degree sequence $d$ all we need is to find a bipartite $\nu$-matched realization of $(d,d)$ for maximum possible $\nu$. Next we describe how to construct a flow graph $G_{d,\nu}$ for any $\nu$, such that $\nu$-matched realization of $(d,d)$ exists if and only if the maximum flow in $G_{d,\nu}$ is $\sum_{i=1}^n d_i - 2\nu$. Moreover, given an integer flow of this value we describe how to construct a desired realization. 

First we describe the construction of $G_{d, \nu}$.
The node set $X$ contains nodes $x_i$ (for $i \in [n]$), corresponding to the vertices of $V$, and the source $s$ is connected to them.
Similarly, the node set $Y$ contains nodes $y_j$ (for $j \in [n]$), 
corresponding to the vertices in $W$, and these are connected to the sink $t$.
The nodes are organized into the following sets:
\begin{itemize}
\item 
Candidates for the matching: 
\begin{align*}
X_M 
& = \{x_i \mid i \in [1,2\nu]\} 
& \text{and} &&
Y_M 
& = \{y_j \mid j \in [1,2\nu]\}
~,
\end{align*}
\item 
Rest of the nodes: 
\begin{align*}
X_R 
& = \{x_i \mid i \in [2\nu+1,n]\}
& \text{and} &&
Y_R 
& = \{y_j \mid j \in [2\nu+1,n]\}
~,
\end{align*}
\item
Source and sink: \{s,t\}.
\end{itemize}
The edges are capacitated and directed from left to right, i.e., an edge $(\alpha, \beta, \gamma)$ leads from the node $\alpha$ to the node $\beta$ and can carry up to $\gamma$ units of flow. The source $s$ has edges leading to every node $x_i$, with capacity $d_i-1$ for $x_i \in X_M$ and $d_i$ for $x_i \in X_R$. Similarly, there are edges leading from every node $y_j$ to the sink $t$, with capacity $d_j-1$ for $y_j \in Y_M$ and $d_j$ for $y_j \in Y_R$. This enforces the degree constraint for the vertices $v_i$ and $w_i$ in $G$. Finally, there's a unit capacity edge $(x_i,y_j,1)$ for every $i, j \in [n]$, such that $i \neq j$ and $j \neq 2\nu - i + 1$. 
Overall, the edge set $\tilde{E}$ of the flow graph is structured as a complete $(X,Y)$ biclique minus an inverted matching and minus diagonal edges, with a source $s$ connected to $X$ and a sink $t$ connected to $Y$. Formally, $\tilde{E}$ (along with the edge weights) is defined as follows:
\begin{align*}
\tilde{E} ~=~
& \set{(s, x_i, d_i-1) \mid i \in [1,2\nu]} \cup \set{(s, x_i, d_i) \mid i \in [2\nu+1,n]} \\
& \cup \set{(x_i, y_j, 1) \mid i, j \in [1,n], ~ i\neq j, j \neq 2\nu - i + 1} \\
& \cup \set{(y_j, t, d_j-1) \mid j \in [1,2\nu]} \cup \set{(y_j, t, d_j) \mid j \in [2\nu + 1,n]}
~.
\end{align*}
The network $G_{d, \nu}$ is schematically depictured in Figure \ref{fig:flow-network}, dashed edges indicate forbidden matching and forbidden diagonal edges. The construction is justified by the following Lemma \ref{lem: matching to flow}.

\tikzset{
  roundNode/.style={
    circle, draw, minimum size=1cm, align=center, font=\Large
  },
  dots/.style={
    draw=none, fill=none, yshift=0.1cm
  },
  edge/.style={
    font=\large
  }
}

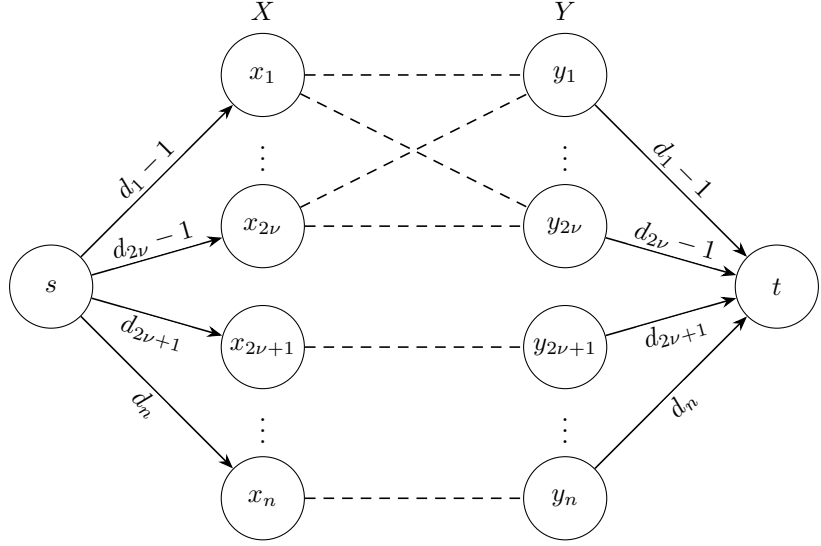
\begin{figure}[t]
\centering
\begin{tikzpicture}[
    font=\small,
    vertex/.style={circle, draw, fill=white,
                   minimum size=11mm, inner sep=1pt},
    arc/.style={-{Stealth[length=2mm]}, semithick},
    forbidden/.style={semithick, dash pattern=on 4pt off 3pt},
    capacity/.style={fill=white, inner sep=1.5pt}
]
    \node[vertex] (s) at (-2.8,0) {$s$};
    \node[vertex] (t) at ( 6.8,0) {$t$};

    \foreach \name/\idx/\yy in {
        a/1/2.8,
        b/{2\nu}/0.8,
        c/{2\nu+1}/-0.8,
        d/n/-2.8
    } {
        \node[vertex] (x\name) at (0,\yy) {$x_{\idx}$};
        \node[vertex] (y\name) at (4,\yy) {$y_{\idx}$};
    }
    
    \node at (0,3.65) {$X$};
    \node at (4,3.65) {$Y$};
    
    \foreach \xx in {0,4} {
        \node at (\xx, 1.8) {$\vdots$};
        \node at (\xx,-1.8) {$\vdots$};
    }

    \foreach \name/\degree in {a/{d_1-1},b/{d_{2\nu}-1}} {
        \draw[arc] (s) --
            node[capacity, above, sloped] {$\degree$} (x\name);
        \draw[arc] (y\name) --
            node[capacity, above, sloped] {$\degree$} (t);
    }

    \foreach \name/\degree in {c/{d_{2\nu+1}},d/{d_n}} {
        \draw[arc] (s) --
            node[capacity, below, sloped] {$\degree$} (x\name);
        \draw[arc] (y\name) --
            node[capacity, below, sloped] {$\degree$} (t);
    }

    \foreach \name in {a,b,c,d}
        \draw[forbidden] (x\name) -- (y\name);

    \draw[forbidden] (xa) -- (yb);
    \draw[forbidden] (xb) -- (ya);
\end{tikzpicture}
\caption{The flow network $G_{d,\nu}$. Solid arcs are labeled by
    their capacities. Dashed lines indicate forbidden pairs.
    The middle arcs $x_i\to y_j$, present exactly when $i\ne j$
    and $i+j\ne2\nu+1$, all have capacity $1$ and are omitted
    for clarity.}
\label{fig:flow-network}
\end{figure}

\begin{lemma}
\label[lemma]{lem: matching to flow}
There exists a $\nu$-matched bipartite graph $\hat{G} = (V, W, \hat{E})$ that realizes the $(d,d)$ 
if and only if the value of a maximum $s-t$ flow in $G_{d,\nu}$ is equal to $\sum_{i=1}^n d_i-2\nu$.
\end{lemma}
\begin{proof}
Suppose the maximum value of an $s-t$ flow in $G_{d, \nu}$ is $\sum_{i=1}^n d_i-2\nu$. Using integral flow theorem, we assume that the maximum flow is integral. Denote a value of the flow from a vertex $x$ to a vertex $y$ by $\FLOW(x,y)$.
Define the bipartite graph $\hat{G} = (V, W, \hat{E})$ as follows:
\begin{enumerate}
\item For each node $x_i \in X$, create a corresponding node $v_i \in V$.
\item For each node $y_j \in Y$, create a corresponding node $w_j \in W$.
\item For every $1\le i, j\le n$, if $\FLOW(x_i,y_j)=1$, then create an edge $(v_i, w_j)$ and add it to $\hat{E}$.
\item  For every $1\le i\le 2\nu$ create an edge $(v_i, w_{2\nu-i+1})$ and add it to $\hat{E}$.
\end{enumerate}

It remains to verify that $\hat{G}$ correctly realizes $(d,d)$ and is $\nu$-matched.
The degree of $v_i$ in $\hat{G}$ equals to the number of edges $(x_i,y_j)$ that carry flow if \\ $i \in [2\nu +1,  n]$ and one more than the corresponding flow if $i \in [2 \nu]$. Since $x_i$ conserves flow, and since all the edges $(s,x_i)$ are saturated, $\deg(v_i) = d_i$. Similarly, $\deg(w_j) = d_j$ for every $j \in [n]$, so $\hat{G}$ satisfies (M1). Since $G_{d,\nu}$ does not have edges $(x_i,y_i)$ for any $i \in [n]$, it follows that $\hat{G}$ does not have edges $(v_i, w_i)$, satisfying (M2). Finally, the matching $\hat{M} = \{(v_i, w_{2\nu-i+1}) \mid 1 \le i \le 2\nu\}$ was added in step 4, so $\hat{G}$ satisfies (M3).

Conversely, suppose $\hat{G} = (V, W, \hat{E})$, where $V = \{v_1, \ldots, v_n\}$ and $W = \{w_1, \ldots, w_n\}$, 
is a $\nu$-matched bipartite graph realizing $(d, d)$ with the matching $\hat{M} = \{(v_i, w_{2\nu-i+1}) \mid 1 \le i \le 2\nu\}$.
We need to define a flow function on $G_{d,\nu}$ and prove that its value 
equals $\sum_{i=1}^n d_i - 2\nu$. 
Since the sum of capacities of outgoing edges from $s$ is exactly this sum, maximum flow cannot exceed it, thus it is enough to construct a flow with the aforementioned value.

\begin{enumerate}
\item \textbf{Saturate Edges from Source and to Sink}:
\begin{align*}
\FLOW(s,x_i) & \gets 
\begin{cases}
d_i - 1 & i \in [1,2\nu], \\
d_i     & i \in [2\nu+1,n].
\end{cases}
&
\FLOW(y_j,t) & \gets 
\begin{cases}
d_j - 1 & j \in [1,2\nu], \\
d_j     & j \in [2\nu+1,n].
\end{cases}
\end{align*}
\item 
\textbf{Flow Through Edges}:
\\
For every $i \in [n]$ and $j \in [n]$ such that $(v_i, w_j) \in \hat{E} \setminus \hat{M}$, 
\\
set $\FLOW(x_i, y_j) \leftarrow 1$,
\item 
\textbf{Completing the flow function}: 
\\
Set the flows of all other edges to $0$.
\end{enumerate}

We need to argue that the defined flow is legal. 
Note that by construction, the total flow out of the source $s$ is $\sum_{i=1}^n d_i - 2\nu$ 
and the total flow into the sink $t$ is $\sum_{j=1}^n d_j - 2\nu$. It remains to verify that 
for every vertex except $s$ and $t$, the incoming and outgoing flows are equal.
\begin{itemize}
\item \textbf{Node $x_i$}: 
By construction, the incoming flow at $x_i$ is $\FLOW(s, x_i) = d_i-1$ if $x_i \in X_M$ and $\FLOW(s, x_i) = d_i$ if $x_i \in X_R$. 
As for the outgoing flow, recall that 
there are exactly $d_i-1$ indices $j$ for which $(v_i,w_j)\in \hat{E} \setminus \hat{M}$ if $i \in [2\nu]$ and $d_i$ 
such indices if $i \in [2\nu+1,n]$. 
Hence, $x_i$ conserves flow.
\item \textbf{Node $y_j$}:
By construction, the outgoing flow at $y_j$ is $\FLOW(y_j, t) = d_j-1$ if $x_j \in X_M$ and $\FLOW(y_j, t) = d_j$ if $y_j \in Y_R$. 
As for the incoming flow, recall that 
there are exactly $d_j-1$ indices $i$ for which $(v_i,w_j)\in \hat{E} \setminus \hat{M}$ if $j \in [2\nu]$ 
and $d_j$ such indices if $j \in [2\nu+1,n]$. 
Hence, $y_j$ conserves flow.
\end{itemize}
The claim follows.
\end{proof}

We conclude the section with the following result.

\begin{theorem}
\label[theorem]{thm:real-MM-fast}
There exists an $O(n^3)$ time algorithm for constructing a realization of a 
given graphic sequence $d$ that also has a matching of maximum size 
(among all possible realizations of $d$).
\end{theorem}
\begin{proof} 
As was shown in the beginning of this section a maximum size of a matching $\nu$ can be found in a running time of 
$O(n \log n)$.
    As for the realization algorithm, given  $\nu$, to find a $\nu$-matched realization of $d$ we first build the flow graph $G_{d, \nu}$ and find a maximum $s$-$t$ flow in it. 
    Since $G_{d, \nu}$ has $O(n)$ vertices and $O(n^2)$ edges, 
    a maximum flow can be computed in $O(n^3)$ time using the maximum flow algorithm by Orlin~\cite{Orlin13}.
    Given a maximum flow, adjacency matrix for $\nu$-matched bipartite realization $\hat{G}$ for the sequence pair $(d,d)$ can be computed in $O(n^2)$ time. 
    Next we go through all the steps of transformation of $\hat{G}$ into $\nu$-matched realization $G$ for $d$ and show that each step takes $O(n^3)$ time.

    \begin{enumerate}[{Step} 1:]
    \item The adjacency matrix of the weighted graph $G^\omega$ can be computed from the adjacency matrix of $\hat{G}$ in $O(n^2)$.
    
    \item It takes $O(n^2)$ to compute the adjacency matrix of $\Ghalf$ using the adjacency matrix of $G^\omega$.
    
    \item Since an Eulerian cycle can be found by linear time in the number of edges and $\Ghalf$ is an even graph, a partition of $\Ghalf$ into disjoint cycles $\calC$ takes $O(n^2)$ time.
        
    \item Separating even cycles in $\calC$ and applying the corresponding modification to them takes $O(n^2)$ time.
          
    \item For every pair of odd cycles $(C, C')$ it takes constant time to find vertices $x \in C$ and $y \in C'$ satisfying \Cref{obs: right x and y for MM} following the proof of \Cref{obs:odd-cycle-condition}. Given $x$ and $y$ it takes linear time $O(|C| +|C'|)$ to apply corresponding modifications. Therefore, the complexity for all pairs of cycles is $O(n^2)$. 
          
    \item At this step $G^\omega$ does not have edges of weight $1/2$, so it can be transformed into $G$ in $O(n^2)$ time.
    \end{enumerate}
The theorem follows.
\end{proof}
\section{Characterization}
\label{sec: characterization}

This section considers the characterization of $\match(d)$ of a given sequence $d$. 
A characterization for this problem was recently given by Erd\H{o}s et al.~{\cite{EKMT24}} and contains $O(n)$ inequalities
(see \Cref{thm:erdos}). 
In previous Sections the realization problem was reduced to a flow problem using the method 
developed by Fulkerson et al.~{\cite{FHM1965}}. 
We show that this technique not only leads to constructive algorithm, but also to an alternative proof of the characterization and show that one of its inequalities is redundant.


\begin{theorem}[Theorem 2.13,~\cite{EKMT24}]
\label[theorem]{thm:erdos}
    For a graphic positive non-increasing degree sequence $d = (d_1, \ldots, d_n)$ 
    and an integer $\nu$ let
    \begin{align*}
        t_d(\nu) & = |\{i>2\nu \mid d_i = d_{2\nu}\}| - |\{i<2\nu \mid d_i = d_{2\nu}\}|, \\
        d'_i &= 
    \begin{cases}
    d_i - 1 & i \leq 2\nu, \\
    d_i     & i > 2\nu.
    \end{cases}
    \end{align*}
    There exists a realization of $d$ with a matching of size $\nu$ if and only if:
    \begin{align*}
        \sum_{i=1}^k d_i &\le k^2 + \sum_{i=k+1}^n \min (d'_i, k), ~ 1 \le k < \nu, ~\\
        \sum_{i=1}^k d_i +|\{i > 2\nu \mid d_i = d_{2\nu} \} | & \le k^2 + \sum_{i=k+1}^n \min (d_i -  I_{d_i = d_{2\nu}}, k), ~ k = 2\nu + t_d(\nu)~.
    \end{align*}
\end{theorem}


For a degree sequence $d$, \Cref{lem: matching to flow,lem: MM equivalence} imply that $\match(d) \ge \nu$ if and only if the maximum $s$-$t$ flow in $G_{d,\nu}$ has value $\sum_{i=1}^n d_i-2\nu$. Let $c(S,T)$ denote the capacity of an $s$-$t$ cut $(S,T)$ in this network. The total capacity leaving $s$ is $\sum_{i=1}^n d_i-2\nu$. Hence, by the max-flow min-cut theorem, the required flow exists if and only if
$c(S,T) \ge \sum_{i=1}^n d_i-2\nu$ for every $s$-$t$ cut.
We show that these cut constraints reduce to the first family of inequalities in \Cref{thm:erdos}. Thus, the additional inequality in that characterization is redundant.

\begin{theorem}
\label[theorem]{thm: characterization}
Let $d$ be a graphic sequence, 
and let $\nu \in [\floor{n/2}]$. 
Define $d'_i = d_i-1$ for $i \in [2\nu]$
and $d'_i = d_i$ for $i \in [2\nu+1,n]$.
There exists a realization of $d$ with a matching of size $\nu$ if and only if
\begin{align}
\label{eq: characterization}
\sum_{i=1}^k d_i 
& \le k^2 + \sum_{i=k+1}^n \min(d'_i,k),
& 1\le k<\nu ~.
\end{align}
\end{theorem}

The case $\nu=0$ is immediate. Throughout the proof, $d$, $\nu$, and $d'$ satisfy the assumptions and definitions of \Cref{thm: characterization}, and an interval with no indices denotes the empty set.
We first rewrite the cut constraints and then reduce the number of inequalities that must be checked.
For an $s$-$t$ cut $(S,T)$, define the index sets
\begin{align*}
A & = \{i\in[n]\mid x_i\in S\}
& \text{and} && 
B & = \{j\in[n]\mid y_j\in S\}
~.
\end{align*}
These sets determine the cut:
$S=\{s\}\cup\{x_i\mid i\in A\}\cup\{y_j\mid j\in B\}$, and $T$ is its complement. Summing the capacities of the edges from $S$ to $T$ gives the following lemma.

\begin{lemma}
\label[lemma]{lem: cuts definition}
Every $s$-$t$ cut $(S,T)$ in $G_{d,\nu}$ satisfies $c(S,T) \ge \sum_{i=1}^n d_i-2\nu$ if and only if
\begin{align}
\label{eq: characterization 1}
\sum_{i\in[n]\setminus A} d'_i + \sum_{j\in B} d'_j + e(A,[n]\setminus B) 
& \ge \sum_{i=1}^n d_i-2\nu,
& \forall A,B\subseteq[n],
\end{align}
where $e(Q,R)$ denotes the number of edges from $\{x_i\mid i\in Q\}$ to $\{y_j\mid j\in R\}$ in $G_{d,\nu}$.
\end{lemma}
For each fixed $A$, we next minimize over $B$ to obtain a single inequality.

\begin{lemma}
\label[lemma]{lem: characterization inequalities}
Inequalities \eqref{eq: characterization 1} hold if and only if
\[
\sum_{i\in A}d_i 
\le |A|(|A|-1)+\sum_{j\in[n]\setminus A}\min(|A|,d_j) -|\{i\in[2\nu]\mid i\notin A,\ 2\nu+1-i\notin A,\ d_i\le|A|\}|
~,
\]
for every $A\subseteq[n]$.
\end{lemma}
\begin{proof}
Since $\sum_{i=1}^n d'_i = \sum_{i=1}^n d_i-2\nu$ 
and $\sum_{i \in A} d'_i + \sum_{i \in [n] \setminus A} d'_i = \sum_{i=1}^n d'_i$,
rearranging \eqref{eq: characterization 1} gives
\begin{align*}
\sum_{i\in A} d'_i-\sum_{j\in B} d'_j
& \le e(A,[n]\setminus B),
&& \forall A,B\subseteq[n]~.
\end{align*}
We first show that it suffices to consider disjoint sets $A$ and $B$.
For arbitrary $A,B\subseteq[n]$, let $A'=A\setminus B$ and $B'=B\setminus A$. Then $A'\cap B'=\emptyset$ and
$\sum_{i\in A}d'_i - \sum_{j\in B}d'_j = \sum_{i\in A'} d'_i - \sum_{j\in B'} d'_j$.
Moreover,
\begin{align}
\label{eq: characterization 2}
e(A',[n]\setminus B')
& =   e(A,[n]\setminus B) - e(A\cap B,[n] \setminus B) + e(A',A \cap B) \notag\\
& =   e(A,[n]\setminus B) - e(A\cap B,[n] \setminus B) + e(A \cap B, A') \notag\\
& =   e(A,[n]\setminus B) - e(A\cap B,[n] \setminus(B\cup A')) \notag\\
& \le e(A,[n]\setminus B)
~,
\end{align}
where the second equality uses the symmetry $e(Q,R) = e(R,Q)$ of the allowed pairs in $G_{d,\nu}$
and the third uses the fact that $A' \cap B = \emptyset$.
Thus, \eqref{eq: characterization 1} holds for all $A,B$ if and only if
\begin{align*}
\sum_{i\in A}d'_i-\sum_{j\in B} d'_j
& \le e(A,[n]\setminus B),
&& \forall A,B\subseteq[n] \text{ such that } A \cap B=\emptyset
~.
\end{align*}
For disjoint $A,B\subseteq[n]$, the definition of $G_{d,\nu}$ gives
\begin{align*}
e(A,[n]\setminus B)
& = |A|(n-1-|B|) - |\set{i\in [2\nu] \mid i \in A \text{ and } 2\nu+1-i\in[n]\setminus B}| \\
& = |A|(n-1-|B|) - |A \cap[2\nu]| + |\set{i \in [2\nu] \mid i \in A \text{ and } 2\nu+1-i\in B}| \\
& = |A|(n-1-|B|) - |A \cap[2\nu]| + |\set{i \in [2\nu] \mid i \in B \text{ and } 2\nu+1-i\in A}|
~.
\end{align*}
Substituting this identity and the definition of $d'$ into the preceding cut inequality yields
\begin{align*}
\sum_{i\in A}d_i-|A\cap[2\nu]|-\sum_{j\in B}d_j+|B\cap[2\nu]|
& \leq |A|(n-1-|B|)-|A\cap[2\nu]| \\
& ~~~~~ + |\set{i\in [2\nu] \mid i \in B \text{ and } 2\nu+1-i\in A}|
~.
\end{align*}
Equivalently,
\[
\sum_{i\in A}d_i
\leq |A|(n-1) + \sum_{j \in B} d_j - |A||B| - |\set{j \in [2\nu] \mid j \in B \text{ and } 2\nu+1-j\notin A}|
~.
\]
It remains to minimize the right-hand side over $B\subseteq[n]\setminus A$.
Let 
\[
\varepsilon_j = 
\begin{cases}
1 & j\in[2\nu] \text{ and } 2\nu + 1 - j \notin A ~, \\
0 & \text{otherwise} ~.
\end{cases}
\]
Including $j$ in $B$ contributes $d_j-|A|-\varepsilon_j$ to the right-hand side; excluding it contributes $0$. 
Hence, minimizing the right hand side is equivalent to
\begin{align*}
& \min_{B\subseteq[n]\setminus A} \left(|A|(n-1)+\sum_{j\in B}(d_j-|A|-\varepsilon_j)\right) \\
& \qquad = |A|(n-1)+\sum_{j\in[n]\setminus A}\min(0,d_j-|A|-\varepsilon_j)\\
& \qquad = |A|(|A|-1)+\sum_{j\in[n]\setminus A}\min(|A|,d_j-\varepsilon_j)\\
& \qquad = |A|(|A|-1)+\sum_{j\in[n]\setminus A}\min(|A|,d_j) - |\set{j\in[2\nu]\mid j\notin A,\ 2\nu+1-j\notin A,\ d_j\le|A|}|
~.
\end{align*}
This is the required inequality.
\end{proof}

The first two terms on the right-hand side in \Cref{lem: characterization inequalities} 
give the usual degree-sum bound in the Erd\H{o}s--Gallai Theorem. 
The subtracted term accounts for the prescribed matching.
Accordingly, define, for every $A\subseteq[n]$,
\begin{align*}
    h(A)&=\sum_{i\in A}d_i-|A|(|A|-1)-\sum_{j\in[n]\setminus A}\min(|A|,d_j),\\
    g(A)&=|\{j\in[2\nu]\mid j\notin A,\ 2\nu+1-j\notin A,\ d_j\le|A|\}|,\\
    f(A)&=h(A)+g(A)~.
\end{align*}
Since $d$ is graphic, the degree-sum bound gives $h(A)\le0$ for every $A$.
By \Cref{lem: characterization inequalities}, the cut constraints are therefore equivalent to $f(A)\le0$ for every $A\subseteq[n]$.

\begin{lemma}
\label[lemma]{lem: characterization prefix}
Under the assumptions of \Cref{thm: characterization}, inequalities \eqref{eq: characterization} imply $f([k])\le 0$, 
for every integer $k \in [n]$.
\end{lemma}
\begin{proof}
We consider three cases.
\begin{enumerate}[{Case} 1:]
\item If $g([k])=0$, then $f([k])=h([k])\le0$, since $d$ is graphic.

\item If $k\ge\nu$, each pair $\{j,2\nu+1-j\}$ with $j\in[2\nu]$ intersects $[k]$. 
      Thus, $g([k])=0$, and Case 1 applies.

\item Suppose $k<\nu$ and $g([k])>0$.
      Some index $j\in[k+1,2\nu-k]$ then satisfies $d_j\le k$. 
      Since $d$ is non-increasing, the last $k$ degrees $d_{2\nu-k+1},\ldots,d_{2\nu}$ are all at most $k$, as well. 
      Consequently, $g([k]) = |\{j\in[2\nu]\setminus[k]\mid d_j\le k\}|-k$.
      Using $\min(d_j-1,k)=\min(d_j,k)-1$ when $d_j\le k$, and $\min(d_j-1,k)=\min(d_j,k)$ otherwise, 
      inequality \eqref{eq: characterization} becomes
\[
\sum_{i=1}^k d_i-k(k-1) - \sum_{i=k+1}^n\min(d_i,k) + |\{j\in [2\nu] \setminus[k] \mid d_j\le k\}| - k
\leq 0
~.
\]
        This is exactly $h([k])+g([k])\le0$.
        \qedhere
    \end{enumerate}
\end{proof}

We can now prove the characterization, namely \Cref{thm: characterization}.

\begin{proof}[Proof of \Cref{thm: characterization}]
Suppose first that $d$ has a realization with a matching of size $\nu$.
By the Inverted Prefix Lemma (\Cref{lem: MM prefix}), there is a realization containing the inverted matching on the first $2\nu$ vertices. Deleting these matching edges gives a realization of $d'$.
By the Erd\H{o}s--Gallai Theorem,
\begin{align}
\sum_{i=1}^k d'_i
& \le k(k-1)+\sum_{i=k+1}^n\min(d'_i,k),
&& 1\le k<\nu
~.
\end{align}
Substituting $d'_i=d_i-1$ for $i\le k$ gives \eqref{eq: characterization}.

Conversely, suppose that \eqref{eq: characterization} holds.
By the flow equivalence and \Cref{lem: cuts definition,lem: characterization inequalities}, it suffices to prove that $f(A)\le0$ for every $A\subseteq[n]$.
Suppose, towards a contradiction, that some set $A$ satisfies $f(A)>0$. Choose such a set $A_*$ of minimum cardinality, and let $k=|A_*|$. Since all degrees are positive, $f(\emptyset)=0$, so $k\ge1$.

First, we show that $d_i\ge k$ for every $i\in A_*$.
Otherwise, choose $j\in A_*$ with $d_j<k$. Then
\begin{align*}
h(A_*\setminus\{j\})-h(A_*)
& =    -d_j + 2(k-1) - \min(d_j,k-1) + |\{i\notin A_*\mid d_i\ge k\}| \\
& \geq |\{i\notin A_*\mid d_i\ge k\}|
~.
\end{align*}
Lowering the cardinality threshold from $k$ to $k-1$ can remove from the set counted by $g$ only indices of degree $k$ outside $A_*$. Removing $j$ from $A_*$ cannot cause any further loss. Hence,
$g(A_*\setminus\{j\})-g(A_*)\ge -|\{i\notin A_*\mid d_i= k\}| \ge -|\{i\notin A_*\mid d_i\ge k\}|$.
It follows that $f(A_*\setminus\{j\})\ge f(A_*)>0$, contradicting the minimality of $A_*$.

Next, suppose that $j\in A_*$ and $i\notin A_*$ satisfy $d_i>d_j$, and let
$A'_*=(A_*\setminus\{j\})\cup\{i\}$.
Since $d_j\ge k$, we have $h(A'_*)-h(A_*)=d_i-d_j\ge1$.
Also, $d_i>k$, so $i$ is not counted by $g(A_*)$.
Replacing $j$ by $i$ can remove at most the index $2\nu+1-i$ from the set counted by $g$. Thus $g(A'_*)-g(A_*)\ge-1$, and $f(A'_*)\ge f(A_*)>0$.
Repeating this replacement, we may assume that the degrees indexed by $A_*$ are $d_1,\ldots,d_k$, 
counted with multiplicities.
In particular, $\{i\in[n]\mid d_i>d_k\}\subseteq A_*$.
Indices of degree $d_k$ may still be chosen arbitrarily, so $A_*=[k]$ need not hold.

We have $d_k\ge k$ and next show that $d_k>k$. 
Suppose that $d_k=k$. Every index outside $A_*$ then has degree at most $k$.
The indices counted by $g(A_*)$ occur in pairs $\{j,2\nu+1-j\}$, so $g(A_*)$ is even.
Moreover,
\[
h(A_*)
= \sum_{i\in A_*} d_i - \sum_{i\notin A_*} d_i - k(k-1)
= \sum_{i=1}^n d_i - 2\sum_{i\notin A_*} d_i - k(k-1)
~.
\]
Since $d$ is graphic, its degree sum is even. Thus $h(A_*)$ is even, and $f(A_*)>0$ implies $f(A_*)\ge2$.
Choose $j\in A_*$ with $d_j=k$. As in the first step,
\begin{align*}
    h(A_*\setminus\{j\})-h(A_*)
    &=-d_j+2(k-1)-\min(d_j,k-1)
      +|\{i\notin A_*\mid d_i\ge k\}|\\
    &=-1+|\{i\notin A_*\mid d_i\ge k\}|~.
\end{align*}
Also, $g(A_*\setminus\{j\})-g(A_*)\ge -|\{i\notin A_*\mid d_i\ge k\}|$.
Therefore, $f(A_*\setminus\{j\})\ge f(A_*)-1>0$, again contradicting minimality.

Let $[q]=\{i\in[n]\mid d_i>d_k\}$ and $[q+1,p]=\{i\in[n]\mid d_i=d_k\}$, and put
$l=|A_*\cap[q+1,p]|=k-q$.
We have $[q]\subseteq A_*\subseteq[p]$.
The value of $h(A_*)$ is unchanged by replacing indices within the block $[q+1,p]$.
Since every degree indexed by $A_*$ exceeds $k$,
\[
g(A_*)=|\{j\in[2\nu]\mid d_j\le k\}|
       -|\{j\in A_*\cap[2\nu]\mid d_{2\nu+1-j}\le k\}|~.
\]
The indices $j\in[2\nu]$ satisfying $d_{2\nu+1-j}\le k$ form a prefix of $[2\nu]$.
Hence, replacing $A_*\cap[q+1,p]$ by $[p-l+1,p]$ can only decrease the second cardinality.
This leaves $h(A_*)$ unchanged and cannot decrease $g(A_*)$ or $f(A_*)$.
We may therefore assume that $A_*=[q]\cup[p-l+1,p]$.

For $0\le r\le p-q$, define $A_r=[q]\cup[p-r+1,p]$, so $A_*=A_l$.
Let $t=\min(d_k-q,p-q)$. Since $d_k>k=q+l$, we have $t\ge l$.
We show that $f(A_t)\ge f(A_l)>0$ by estimating the successive differences.
For $1\le r\le t$, the new index $p-r+1$ has degree $d_k>q+r-1$ and is not counted by $g(A_{r-1})$.
Adding it can remove at most its partner $2\nu+1-(p-r+1)$ from the set counted by $g$.
Increasing the threshold from $q+r-1$ to $q+r$ can add only indices of degree $q+r$. Consequently,
\begin{align*}
    -1\le g(A_r)-g(A_{r-1})
    \le|\{j\in[n]\mid d_j=q+r\}|~.
\end{align*}
For the same range of $r$, put $m=q+r$. Every index in $A_r$ has degree at least $d_k\ge m$, and therefore
\begin{align*}
    h(A_r)-h(A_{r-1})
    &=d_k-2(m-1)+\min(d_k,m-1)
      -|\{j\notin A_r\mid d_j\ge m\}|\\
    &=d_k-(m-1)-|\{j\notin A_r\mid d_j\ge m\}|\\
    &=d_k+1-|\{j\in[n]\mid d_j\ge m\}|~.
\end{align*}
Applying the upper bound on $g(A_l)-g(A_{l-1})$ gives
$f(A_l)-f(A_{l-1})\le d_k+1-|\{j\in[n]\mid d_j\ge k+1\}|$.
Hence, for $l+1\le r\le t$,
\begin{align*}
    f(A_r)-f(A_{r-1})
    &\ge d_k-|\{j\in[n]\mid d_j\ge q+r\}|\\
    &\ge d_k-|\{j\in[n]\mid d_j\ge k+1\}|\\
    &\ge f(A_l)-f(A_{l-1})-1\\
    &\ge0~.
\end{align*}
Here $l=k-q\ge1$, and minimality gives $f(A_{l-1})\le0<f(A_l)$.
Since these values are integers, $f(A_l)-f(A_{l-1})\ge1$.
Summing the differences yields $f(A_t)\ge f(A_l)>0$.

Finally, consider two cases.
\begin{enumerate}[{Case} 1:]
    \item If $d_k\ge p$, then $t=p-q$ and $A_t=[p]$.
    By \Cref{lem: characterization prefix}, $f(A_t)=f([p])\le0$, a contradiction.
    \item If $d_k<p$, the preceding bound at $r=l$ gives
\[
f(A_l)-f(A_{l-1})
\leq d_k+1-|\{j\in[n]\mid d_j\ge k+1\}|
\leq d_k+1-p
\leq 0
~.
\]
    The second inequality holds because all $p$ indices in $[p]$ have degree at least $d_k>k$.
    This again contradicts $f(A_l)-f(A_{l-1})\ge1$.
\end{enumerate}
Thus $f(A)\le0$ for every $A\subseteq[n]$, completing the proof.
\end{proof}
\section*{AI Disclosure}
Chat-GPT-6 Astra was used to generate an initial version of the proof of Theorem \ref{thm: characterization}, which was then simplified and rewritten by us. The tool was also used to help with the drawing of Figure \ref{fig:flow-network}, for spell checking and to enhance readability throughout the paper. We verified the correctness and originality of all content including references.




\bibliographystyle{abbrv}
\bibliography{realizations}

@string{DAM = "Discrete Appl. Math."}

@string{DAM = "DAM"}

@string{DAM = "Discr.\ Appl.\ Math."}

@string{DM = "Discrete Math."}

@string{DM = "DM"}

@string{DM = "Discr.\ Math."}

@string{LNCS = "LNCS"}

@article{Anstee85,
  author       = {Richard P. Anstee},
  title        = {An Algorithmic Proof of {Tutte's} $f$-Factor Theorem},
  journal      = {J.\ Alg.},
  volume       = {6},
  number       = {1},
  pages        = {112--131},
  year         = {1985},
}

@article{Anstee90,
  author       = {Richard P. Anstee},
  title        = {Simplified existence theorems for $(g, f)$-factors},
  journal      = DAM,
  volume       = {27},
  number       = {1-2},
  pages        = {29--38},
  year         = {1990},
}

@inproceedings{BKPR26-domset,
  author       = {Amotz Bar-Noy and Igor Kalinichev and David Peleg and Dror Rawitz},
  title        = {Degree Realization with Minimum Dominating Set},
  booktitle    = {27th IPCO},
  series       = LNCS,
  volume       = {16588},
  pages        = {64--80},
  year         = {2026},
  doi          = {10.1007/978-3-032-28691-8\_5},
}

@article{Barrus16,
    title = {On realization graphs of degree sequences},
    journal = {Discrete Mathematics},
    volume = {339},
    number = {8},
    pages = {2146-2152},
    year = {2016},
    issn = {0012-365X},
    doi = {https://doi.org/10.1016/j.disc.2016.03.012},
    url = {https://www.sciencedirect.com/science/article/pii/S0012365X16300577},
    author = {Michael D. Barrus},
}

@article{BockRautenbach19,
  author       = {Felix Bock and Dieter Rautenbach},
  title        = {On matching numbers of tree and bipartite degree sequences},
  journal      = DM,
  volume       = {342},
  number       = {6},
  pages        = {1687--1695},
  year         = {2019},
}

@article{Chungphaisan78,
  author       = {V. Chungphaisan},
  title        = {Construction of hamiltonian graphs and bigraphs with prescribed degrees},
  journal      = {J. Comb. Theory {B}},
  volume       = {24},
  number       = {2},
  pages        = {154--163},
  year         = {1978},
}

@Article{EG60,
  author = 	 {Paul Erd\H{o}s and Tibor Gallai},
  title = 	 {Graphs with Prescribed Degrees of Vertices {[Hungarian]}},
  journal = 	 {Matematikai Lapok},
  year = 	 {1960},
  volume = 	 {11},
  pages = 	 {264--274},
}

@article{EKM13,
  author    = {Peter L. Erd{\H{o}}s and Zoltan Kiraly and Istvan Miklos},
  title     = {On the Swap-Distances of Different Realizations of a Graphical Degree Sequence},
  journal   = {Combinatorics, Probability and Computing},
  volume    = {22},
  pages     = {366--383},
  year      = {2013}
}

@article{EKMT24,
   title={New Results on Graph Matching from Degree-Preserving Growth},
   author={Erdős, Péter L. and Kharel, Shubha R. and Mezei, Tamás Róbert and Toroczkai, Zoltán},
   volume={12},
   number={22},
   article-number = {3518},
   journal={Mathematics},
   year={2024},
   DOI = {10.3390/math12223518}
}

@article{Fulkerson1960,
title = "Zero-one matrices with zero trace",
journal = "Pacific J.\ Math.",
volume    = {10},
pages = "831--836",
year = "1960",
author = "D.R. Fulkerson"
}

@article{FHM1965,
    title={Some Properties of Graphs with Multiple Edges},
    author={D. R. Fulkerson and A. J. Hoffman and M. H. McAndrew},
    journal={Canad.\ J.\ Math.},
    volume={17},
    pages={166--177},
    year={1965},
}

@article{GHR16,
  author       = {Michael Gentner and Michael A. Henning and Dieter Rautenbach},
  title        = {Largest domination number and smallest independence number of forests
                  with given degree sequence},
  journal      = DAM,
  volume       = {206},
  pages        = {181--187},
  year         = {2016},
}

@article{GHR18,
  author       = {Michael Gentner and Michael A. Henning and Dieter Rautenbach},
  title        = {Smallest domination number and largest independence number of graphs
                  and forests with given degree sequence},
  journal      = {J. Graph Theory},
  volume       = {88},
  number       = {1},
  pages        = {131--145},
  year         = {2018},
}

@inproceedings{GJL99,
  author    = {R. J. Gould and M. S. Jacobson and J. Lehel},
  title     = {Potentially {$G$}-graphical degree sequences},
  booktitle = {Combinatorics, Graph Theory, and Algorithms},
  volume = {1},
  pages     = {451--460},
  year      = {1999},
}

@article{Hakimi62,
    author={Seifollah Louis Hakimi},
    title={On Realizability of a Set of Integers as Degrees of the Vertices of a Linear Graph -- {I}},
    journal = {SIAM J. Appl. Math.},
    volume={10},
    number = {3},
    pages={496--506},
    year={1962},
    publisher={Society for Industrial and Applied Mathematics}
}

@Article{havel55,
  author={Havel, V\'{a}clav},
  title = 	 {A Remark on the Existence of Finite Graphs},
  journal = 	 {Casopis Pest. Mat.},
  year = 	 {1955},
  volume = 	 {80},
  pages = 	 {477--480},
}

@inproceedings{KezdyLehel96,
  author    = {Andr\'e K\'ezdy and Jen\H{o} Lehel},
  title     = {Degree sequences of graphs with prescribed clique size},
  booktitle = {Graph Theory, Combinatorics, Algorithms, and Applications},
  pages     = {535--544},
  year      = {1999},
}

@article{KleitmanW73,
author    = {Daniel J. Kleitman and D. L. Wang},
title    = {Algorithms for constructing graphs and digraphs with given valences and factors},
journal   = DM,
volume    = {6},
pages     = {79--88},
year      = {1973}
}

@article{Kundu73,
  author       = {Sukhamay Kundu},
  title        = {The $k$-factor conjecture is true},
  journal      = DM,
  volume       = {6},
  number       = {4},
  pages        = {367--376},
  year         = {1973},
}

@article{Kundu74,
  author       = {Sukhamay Kundu},
  title        = {Generalizations of the $k$-factor theorem},
  journal      = DM,
  volume       = {9},
  number       = {2},
  pages        = {173--179},
  year         = {1974},
}

@article{Kundu74a,
  author       = {Sukhamay Kundu},
  title        = {A factorization theorem for a certain class of graphs},
  journal      = DM,
  volume       = {8},
  number       = {1},
  pages        = {41--47},
  year         = {1974},
}

@inproceedings{Orlin13,
  author       = {James B. Orlin},
  title        = {Max flows in {$O(nm)$} time, or better},
  booktitle    = {45th STOC},
  pages        = {765--774},
  year         = {2013},
}

@inproceedings{Rao79,
author    = {A. R. Rao},
title     = {The clique number of a graph with a given degree sequence},
Booktitle = {ISI Lect. Notes},
year      = {1979},
volume    = {4},
pages     = {251--267},
}

@article{RaoRao72,
title = {On factorable degree sequences},
journal = {Journal of Combinatorial Theory, Series B},
volume = {13},
number = {3},
pages = {185-191},
year = {1972},
issn = {0095-8956},
doi = {https://doi.org/10.1016/0095-8956(72)90055-X},
url = {https://www.sciencedirect.com/science/article/pii/009589567290055X},
author = {A {Ramachandra Rao} and S.B Rao},
}

@article{Yin11,
  author       = {Jian-Hua Yin},
  title        = {A Rao-type characterization for a sequence to have a realization containing a split graph},
  journal      = DM,
  volume       = {311},
  number       = {21},
  pages        = {2485--2489},
  year         = {2011},
}

@article{Yin12,
  author       = {Jian-Hua Yin},
  title        = {A short constructive proof of {A. R. Rao's} characterization of potentially {$K_{r+1}$}-graphic sequences},
  journal      = DAM,
  volume       = {160},
  number       = {3},
  pages        = {352--354},
  year         = {2012},
}

\end{document}